\documentclass[abstracton, DIV=14, 11pt,letterpaper]{article}
\usepackage[margin=1in]{geometry}
\usepackage{times}  
\usepackage{helvet} 
\usepackage{courier}  
\usepackage[hyphens]{url}  
\usepackage{graphicx} 
\usepackage{natbib}  
\usepackage{caption} 
\usepackage{amsfonts}
\usepackage{amsthm}
\usepackage{amsmath}
\usepackage{xcolor}
\usepackage{authblk}
\usepackage{thmtools}
\usepackage{thm-restate}

\usepackage{graphicx}
\usepackage{comment}
\usepackage[ruled,linesnumbered,norelsize]{algorithm2e}
\usepackage[switch]{lineno}  %
\usepackage{color}
\usepackage{url}
\usepackage[colorlinks=true,bookmarks=true, pdftex, citecolor = black, linkcolor = black, urlcolor = black]{hyperref}
\hypersetup{%
	breaklinks,
	baseurl       = http://,
	pdfborder     = 0 0 0,
	pdfpagemode   = UseNone,
	pdfpagelabels = false,
	pdfstartpage  = 1,
	bookmarksopen = true,
	bookmarksdepth= 3,
}
\usepackage{authblk}

\usepackage{natbib}
\usepackage{amsfonts}
\newtheorem{definition}{Definition}

\usepackage{xcolor}         

\usepackage{color}
\definecolor{english}{rgb}{0.0, 0.5, 0.0}

\usepackage{subfigure}
\usepackage{amsmath}

\begin{document}
	
\title{Learning to Mitigate AI Collusion on Economic Platforms\footnote{Author order is alphabetical. This research is funded in part by Defense Advanced Research Projects Agency under Cooperative Agreement HR00111920029. The content of the information does not necessarily reflect the position or the policy of the Government, and no official endorsement should be inferred. This is approved for public release; distribution is unlimited. The work of G. Brero was also supported by the SNSF (Swiss National Science Foundation) under Fellowship P2ZHP1\_191253. We thank Emilio Calvano and Justin Johnson for their availability to answer questions about their work and for guidance in replicating some of their results. We also thank Alon Eden, Matthias Gerstgrasser, and Alexander MacKay for for helpful discussions and feedback.}}

\author[1]{Gianluca Brero}
\author[2]{Nicolas Lepore}
\author[1]{Eric Mibuari}
\author[1]{David C. Parkes}
\affil[1,2]{John A. Paulson School of Engineering and Applied Sciences, Harvard University}
\affil[1]{\small \texttt {\{gbrero,mibuari,parkes\}@g.harvard.edu}}
\affil[2]{\small \texttt {nlepore33@gmail.com}}

\date{\today}

\begin{titlepage}

\maketitle

\begin{abstract}
	{Algorithmic pricing on online e-commerce platforms raises the concern of tacit collusion, where reinforcement learning algorithms learn to set collusive prices in a decentralized manner and through nothing more than profit feedback. This raises the question as to whether collusive pricing can be prevented through the design of suitable ``buy boxes," i.e., through the design of the rules that govern the elements of e-commerce sites that promote particular products and prices to consumers. 
	In this paper, we demonstrate that reinforcement learning (RL) can also be used by platforms to learn buy box rules that are effective in preventing collusion by RL sellers.
	For this, we adopt the methodology of {\em Stackelberg POMDPs}, and demonstrate success in learning robust rules that continue to provide high consumer welfare together with sellers employing different behavior models or having out-of-distribution costs for goods.}
\end{abstract}

\end{titlepage}

\section{Introduction}

The last decade has witnessed a dramatic shift of trading from retailers to online e-commerce platforms such as Amazon and Alibaba. In these platforms, sellers are increasingly using algorithms to set prices.  Algorithmic pricing can be beneficial for market efficiency, enabling sellers to quickly react to market changes and also in enabling price competition. At the same time,  the U.S.~Federal Trade Commission (FTC) \cite{FTC} and  European Commission \citep{OECD} have raised 
concerns that algorithmic pricing may facilitate collusive behaviors. 
\citet{calvano2020artificial} support these concerns through a study of pricing agents in a simulated platform economy, and show  that commonly used reinforcement-learning (RL) algorithms  learn to initiate and sustain collusive behaviors. 
\citet{assad2020algorithmic} also provide empirical support for algorithmic collusion in a study of  Germany’s retail gas stations, showing an association between  algorithmic pricing and an increase in price markups. As highlighted by~\citet{calvano2020protecting}, these kinds of collusive behaviors are unlikely to be a violation of antitrust law, as they are learned responses to profit signals and not the result of explicit agreements.

One can try to prevent algorithmic collusion by introducing suitable rules by which platforms can choose which sellers to promote to buyers, thus promoting competition. 
Could Amazon's ``buy box algorithm," for example, play this role in the future, in determining for  a given consumer search which products and prices to highlight to a consumer?
Responding to this,~\citet{johnson2020platform} design  hand-crafted rules that succeed in  hindering collusion between RL algorithms. At the same time, their rules introduce the undesirable effect of limiting consumers to a single seller, and there remains  potential for more effective interventions.

In this paper, we demonstrate for the first time  how   RL  can also  be used  defensively by a platform 
to automatically design  rules that promote consumer welfare and prevent collusive pricing.  This is a problem of multi-agent learning, with the interaction between the platform and sellers modeled as a \textit{Stackelberg game} \citep[]{fudenberg1991game}. The  leader is the platform designer and sets the platform rules and the sellers respond,  using RL to set prices given these rules.
We introduce the class of \textit{threshold platform rules}, and formally show  that this class contains rules that approximately maximize consumer surplus in a unique subgame perfect equilbrium.
{At the same time, this class  of threshold rules is fragile to unexpected deviations by sellers, for example caused by cost perturbations. The role of RL on the part of the platform is to learn rules with similar performance that are also more robust.} 

To solve the Stackelberg problem, we make use of the  {\em Stackelberg partially observable Markov decision process} (POMDP) framework~\citep{brerolearning}, which defines an episode structure of a POMDP such that the RL algorithm representing the leader will learn to optimize its reward (here, consumer surplus) given that its rules cause re-equilibration on the part of the followers (here, the sellers who  use Q-learning algorithms to set prices). {The Stackelberg POMDP framework is well formed as long as the re-equilibration behavior of sellers can be modeled through Markovian dynamics, as is the case with Q-learning.}

We show successful results in learning effective platform policies that outperform handcrafted rules~\citep{johnson2020platform}. This demonstrates how the Stackelberg POMDP framework can be successfully applied in settings where followers play repeated games, and their strategies are also policies trained via reinforcement learning algorithms. We then show how our threshold platform rules allow us to obtain a similar learning performance when training the platform policy ``in the wild,'' i.e., without accessing the sellers' private information. With this, we demonstrate how the Stackelberg POMDP framework can be applied in more general learning scenarios than the offline learning ones for which it was originally designed. Finally, we show how the platform rules learned via our Stackelberg POMDP framework continue to be effective when market conditions change, for example as the result of a change to the cost structure of sellers.

\paragraph{Further related work.}  
\citet{zheng2022,tang2017reinforcement,shen2020reinforcement,brero2021reinforcement}  make use of  RL to optimize different economic systems (including matching markets, internet advertising, tax policies, and auctions) under strategic agents' responses. Unlike our work, these methods do not leverage the designer's commitment power or  the Stackelberg structure of the induced game.~\citet{brerolearning}   introduce  and study the  Stackelberg POMDP framework for a very different setting than that of the present paper: the design of sequential price auctions.\footnote{The only other method we know for Stackelberg learning  in stochastic games with provably guarantees
solves for a single follower~\citep{MishraVV20}; see also~\citet{MguniJSMCC19,cheng17,ShiYWWZLA20} and~\citet{ShuT19}, and~\citet{thara07}  for a partial convergence result for a static game with two followers.
For other convergence results for single-follower, static, and often zero-sum games see~\citet{li_robust_2019, sengupta_multi-agent_2020,xu2021robust,FiezCR20,JinNJ20}. 
For multi-follower static games,~\citet{wang22} make  use of a differentiable relaxation of follower best-response behavior together with a subroutine to solve an optimization problem for follower behavior.}

\citet{abada2020artificial} study collusion by RL pricing in markets for electric power, and  use machine learning by a regulator agent for the mitigation of collusion, albeit without a Stackelberg framing (and without success, leading to lower welfare than the collusive outcome).
The broader research program on {\em differentiable economics}  uses representation learning for optimal economic design~\citep{DuettingFNPR19,STZ18,kuo20,tacchetti19,rahme21,curry22,rahmeJBW21, curry20, peri21}; this work  avoids the need for Stackelberg design by emphasizing the use of direct, incentive-compatible mechanisms. 
Also related is {\em empirical mechanism design}~\citep{ViqueiraCMG19,VorobeychikKW06,BrinkmanW17}, which applies
empirical game theory to search for the equilibria of mechanisms with a  set of candidate strategies~\citep{Wellman06,KiekintveldW08,JordanSW10}; see also~\citet{bunz2020designing} for the design of iterative auctions.

\section{Preliminaries}
\label{sec:preliminaries}

\paragraph{Seller Competition Model.}
\label{sec:econ_model}
There is a set of sellers $\mathcal N=\{1,\ldots, n\}$, each of whom sells a {differentiated} product on an economic platform. Each seller has the same (private) {\em marginal cost} $c>0$ for producing one unit of its product. Sellers interact with each other repeatedly over time in setting prices and selling goods. At each time period, $t=0,1,\dots$, each seller $i$ observes all past prices, and sets a {\em price} $p_{i,t}\ge 0$. We let $p_t=(p_{1,t},\ldots,p_{n,t})$ denote a generic price profile quoted at time $t$. 
The platform sets the rules of a buy box that govern, in each period $t$, which set $\mathcal N_t\subseteq \mathcal N$ of   sellers  are available.   Consumers  can only buy from these sellers and others  forfeit sales. There is also an outside option, indexed by $0$, which provides each consumer with a fallback choice with zero utility.

Following~\citet{johnson2020platform},  competition between sellers for consumer demand is modeled through the standard {\em logit model} of consumer choice.
For this, seller $i$ has {\em quality index} $\alpha_i>0$, this {providing horizontal differentiation} across products, and the outside good has quality index $\alpha_0>0$.
%
In the logit model, each  consumer samples  $\zeta_0, \zeta_1, ... \zeta_n$, independently from a type I extreme value distribution with {\em scale parameter} $\mu>0$, for each  product and the outside option, with utility $\alpha_i + \zeta_i - p_{i,t}$ for product $i$, and $\alpha_0 + \zeta_0$ for the outside option. Considering
a unit mass of consumers in period $t$, seller $i \in \mathcal N_t$ receives fractional demand  $D_i(p_t; \mathcal{N}_t) = 	\mathrm{exp}((\alpha_i-p_{i,t})/\mu)/ \lambda(p_t;\mathcal N_t)$, where $\lambda(p_t;\mathcal N_t) = \sum_{j\in \mathcal N_t} \mathrm{exp}((\alpha_j-p_{j,t})/\mu) +  \mathrm{exp}(\alpha_0/\mu)$, and any seller $i\notin \mathcal N_t$ has zero demand. Scale parameter $\mu>0$ serves to control the extent of horizontal differentiation, with no differentiation and perfect substitutes  obtained  as $\mu \to 0$. The total {\em  consumer surplus}  is $U(p_t; \mathcal{N}_t) = \mu \cdot \log[\lambda(p_t;\mathcal N_t)]$, and is maximized with minimum prices and  all sellers  displayed (so consumers have a full choice of products). 
Seller $i$'s {\em profit} $\rho_i$ in period $t$ is  $\rho_{i}(p_t; \mathcal{N}_t)=(p_{i,t} - c)\cdot D_i(p_t; \mathcal{N}_t)$, and its per-unit profit multiplied by demand.

\paragraph{Reinforcement learning by sellers.}\label{sec:MDPs} 
{In a single-agent Markov decision process (MDP)}, an agent {faces a sequential decision problem under uncertainty}. At each step $t$, the agent observes a state variable $s_t\in S$ and chooses an action $a_t \in A$. Upon action $a_t$ in state $s_t$, the agent obtains  reward $r(s_t,a_t)$, and the environment moves to state $s_{t+1}$ according to 
$p(s_{t+1}|s_t,a_t)$. We let $\tau=(s_0,a_0,...,s_T,a_T)$ denote a state-action trajectory determined by executing policy  policy $\pi:S\to A$, and $p_\pi(\tau)$ denote the probability of trajectory $\tau$.
The optimal policy $\pi^*$ solves 
$$ \pi^*\in \text{argmax}_\pi E_{\tau\sim p_\pi(\tau)} \left[\sum_{t=0}^T \delta^t r(s_t, a_t)\right]\!,$$
where $\delta\in[0,1]$ is the discount factor and time-horizon $T$ can be finite or infinite. In 
a {\em partially-observable MDP (POMDP)},  the policy $\pi$ cannot access state $s_t$   but only observation $o_t$ sampled from  $p(o_t|s_t)$. 
%
A {\em multi-agent MDP}~\citep{boutilier1996planning} for $n$ agents has states $S$ common to all agents and a set of actions $A_i$ for each agent $i$. When each agent $i$ picks action $a_{i,t}$ in state $s_t$, the environment moves to state $s_{t+1}$ according to a distribution $p(s_{t+1}|s_{t}, a_{1,t},..,a_{n,t})$ and  agent $i$ obtains a reward $r_i(s_t, a_{t})$ that depends on the joint action.
We follow \citet{calvano2020artificial} and \citet{johnson2020platform} and  adopt decentralized Q-learning by sellers   as a positive theory for  sellers in regard to their behavior in setting prices on an e-commerce platform (see Appendix~\ref{app:Qlearning}). Although Q-learners may not converge,  we also confirm these earlier  studies in showing convergence in our simulations (defined over a particular time horizon as detailed by~\citet{johnson2020platform}). 

\section{The Platform Stackelberg Problem}
\label{sec:prob_description}

To formalize the problem facing the platform in mitigating collusive behavior by sellers, we  model the interaction between the platform, which sets the rules of the buy box, and  the sellers as a {\em Stackelberg game}.
The platform designer is the leader, and fixes the platform rules. The sellers are the followers, and  play an infinitely repeated game according to these rules. As discussed above, and following~\citet{calvano2020artificial} and~\citet{johnson2020platform}, we model the sellers' behavior through decentralized Q-learning. As a result, the problem facing the platform is a  {\em behavioral Stackelberg problem}, in that the followers are modeled as  Q-learners (and need not, necessarily, be playing an equilibrium of the induced game). 

\paragraph{The sellers.}
In this model, we fix the states that comprise the MDP of a seller to include the  prices set by all sellers in the last period, i.e., $s_t = p_{t-1}$.
We initialize $s_0$  to be a randomly selected price profile. 
The action of a seller is modeled as one of $m$ equally-spaced points in the interval ranging from just below the sellers’ cost $c$ to above the monopoly price when no buybox is used (as in \citet{johnson2020platform}).
At each step $t\geq 0$, each seller $i$ selects a  price $p_{i,t}$ and is rewarded by its per-period profit $\rho_{i}(p_t;\mathcal N_t)$, which  depends on $p_t=(p_{1,t},\ldots,p_{n,t})$ and the choice of which sellers $\mathcal N_t$ are displayed by the   platform.

\paragraph{The platform.}
To formalize the platform's problem, let $\sigma^* = (\sigma_{1}^*,..,\sigma_{n}^*)$ denote a  strategy profile selected by Q-learning on the part of sellers, in response to the platform rule, and {in the long run,} after a suitably large number of steps. We leave implicit here the dependence of the sellers' strategy profile on the platform's policy.
The platform  must decide in each period which sellers to display to consumers. 
For this, we denote the platform rule as {\em policy} $\pi$, and {we adopt for the state of the platform policy the  prices quoted by sellers in step $t$, $p_t$}, so that the platform's policy uses these prices to select a set  of agents to display, with $\mathcal N_t$ selected according to $\pi(p_t)$.
Let $p^*_{t} = \sigma^*(s_{t})$ denote a  price profile chosen under seller strategies $\sigma^*$, {i.e., in response to the platform rules,} and at {some large enough time  step} $t^*$, and let $\tau^* = (p^*_{t^*},p^*_{t^*+1},..)$ denote a trajectory of  prices forward from $t^*$.
We denote the distribution of these trajectories as $p_\pi(\tau^*)$. As above, the dependence on the platform's policy is left implicit in this notation.

The  Stackelberg problem facing the platform  is to find  a platform policy $\pi$ that maximizes consumer surplus given the effect of this policy on the induced strategy profile of sellers.
\begin{definition}[Behavioral Stackelberg Problem]\label{def:behavioralStack}
The optimal platform policy solves $\pi^* \in \text{argmax}_{\pi} \mathit{CS}(\pi)$, where $\mathit{CS}(\pi)$ is the expected sum  consumer surplus when sellers follow strategy $\sigma^*$ forward from period $t^*$, 
	\begin{equation}\label{eq:platform_obj}
	\mathit{CS}(\pi) = \mathbb{E}_{\tau^*\sim p_\pi(\tau^*)}\left [ \sum_{t=t^*}^{T^*} U(p^*_t; \pi(p^*_t)) \right ]\!\!,
	\end{equation}
    where $T^*$ is suitably chosen horizon  and $p_\pi(\tau^*)$ denotes the distribution over Q-learning induced sellers' pricing trajectories, in response to platform policy $\pi$.
\end{definition}

\section{Learning Optimal Platform Rules}
\label{sec:learning_problem}

{In this section, we solve the platform's problem, in responding to Q-learning sellers, through the {\em Stackelberg POMDP} framework~\citep{brerolearning}. This creates a suitably defined POMDP} in which the optimal policy solves the behavioral Stackelberg problem (Definition~\ref{def:behavioralStack}).
\begin{definition}[Stackelberg POMDP for platform rules]
	The {\em Stackelberg POMDP {for  platform rules}} is a finite-horizon POMDP, where each episode has the following two  phases: 
	
	%
		$\bullet$ An {\em equilibrium phase}, consisting of $n_e\geq 1$ steps. In this phase, each state $s_t$ includes the step counter $t$, the sellers' current Q-matrices, and the prices $p_t$ quoted by the agents. Observations consists of the prices quoted by the sellers ($o_t=p_t$) and policy actions determine the set of agents displayed (in their more general version, $a_t = \mathcal N_t$). State transitions are determined by Q-learning, where each seller $i$ updates its Q-matrix after being rewarded by $\rho_i(p_t;\mathcal N_t)$. The policy has zero reward in every time step ($r(s_t,a_t) = 0$, for $t\le n_e$). 
		
		%
		$\bullet$ A {\em reward phase}, consisting of $n_r\geq 1$ steps, each with the same actions, states, and observations as the equilibrium steps. The reward phase differs in two ways. First, the Q-matrices of sellers are not updated, and second, the platform policy now receives a non-zero reward, and this is set in each step to be equal to the consumer surplus in that step ($r(s_t,a_t) = U(p_t;\mathcal N_t)$, for $t>n_e$).
	%
\end{definition}

This Stackelberg POMDP formulation is an adaptation of that provided by~\citet{brerolearning}, who used it to learn sequential price mechanisms (SPMs) in the presence of communication from bidders.
Here, our stage games replace SPMs,  and the followers respond through Q-learning dynamics rather than no-regret algorithms. 
{Following~\citet{brerolearning},} we show the Stackelberg POMDP formulation is well-founded by showing that an optimal policy will also solve the Behavioral Stackelberg design problem of Definition~\ref{def:behavioralStack}. Specifically, when the number of reward steps $n_r$ is large enough and when $n_e\ge t^*$, the optimal policy, denoted $\pi^*_{n_e,n_r}$, for the Stackelberg POMDP with $n_e$ equilibrium and $n_r$ reward steps maximizes the objective in Equation~\eqref{eq:platform_obj}. 
\begin{restatable}{proposition}{propoptimalpolicy}\label{prop:optimal_policy}
	The optimal Stackelberg POMDP policy $\pi^*_{n_e,n_r}$, for an equilibrium phase with $n_e\geq 1$ steps and a reward phase with $n_r\geq 1$ steps, maximizes $\mathit{CS}(\pi)$, for seller behavior after $n_e$ steps when $n_r=T^*$.
\end{restatable}
\begin{proof}
Let $\tau \sim p_\pi(\tau) $ denote a generic trajectory determined by executing policy $\pi$ in the Stackelberg POMDP environment.
	We have 
	\begin{equation}\label{eq:reward_steps_obj}
	\pi^*_{n_e,n_r} \in \text{argmax}_\pi \mathbb{E}_{\tau\sim p_{\pi}(\tau)}\left[ \sum_{t={n_e+1}}^{n_e+n_r} r(s_t,a_t) \right],
	\end{equation}
	recognizing $r(s_t,a_t)=0$ if $t\le n_e$.
	After replacing $r(s_t,a_t)$ with $U(p_t;\pi(p_t))$, the objective in~\eqref{eq:reward_steps_obj} is 
	\begin{equation*}
		\mathbb{E}_{\tau\sim p_{\pi}(\tau)}\left[\sum_{t=n_e+1}^{n_e+n_r} U(p_t;\pi(p_t))\right],
	\end{equation*}
	which is equal to~\eqref{eq:platform_obj} when $n_r=T^*$.
\end{proof}



{\citet{brerolearning} use the Stackelberg POMDP framework  in an ``offline'' environment, i.e., in a simulation that assumes  access, at design time, to followers' internal information. This allows them to solve their POMDP using the paradigm of {\em centralized training and decentralized execution} \citep{lowe2017multi}. The leader policy is trained via an actor-critic deep RL algorithm, and the critic network (which estimates the sum of rewards until the end of the episode) accesses the full state during training. Only the actor network, which represents the policy, is restricted to the partial-state information. 

Here, we   also study the use of the Stackelberg POMDP framework to train useful leader policies ``in the wild,'' where the learning algorithm of the platform can only access the kind of information that an economic platform would have based on observations of sellers. As we will empirically demonstrate, we can successfully relax the offline learning requirements---i.e., we operate without (1)  access to sellers' private information in regard to Q-matrices and exploration rate, and (2) requiring that the Q-matrices of sellers become frozen for the reward phase of the Stackelberg POMDP---without affecting learning performance.}\footnote{We notice this is also in line with the recent findings in \citet{fujimoto2022should} highlighting how the Bellman error minimization (for which we require environments to be Markovian) may not be a good proxy of the accuracy of the value function.}

\paragraph{Threshold platform rules.} In our experiments, we consider the class of \textit{threshold platform rules}. These  rules  use the current prices set by sellers to set a price threshold above which a seller will not be displayed, with the same threshold set for  all sellers.
%
\begin{definition}[Threshold Platform Rule]
	A {\em threshold platform rule} sets a threshold $\tau(p_t)\geq 0$, for
	each price profile $p_t$, 
	such that $\mathcal N_t = \{i\in \{1,..,n\}: p_{i,t} \le \tau(p_t)\}$, i.e., any seller whose price is no greater than the threshold is displayed to consumers.
\end{definition}

This class of threshold rules has a corresponding  optimality result: there is a threshold rule that makes {the market competitive, with all sellers  displayed and consumer surplus  maximized} in the unique subgame perfect Nash equilibrium (SPE) of the induced continuous pricing game. 
Even though the pricing behaviors that arise from Q-learning  need not converge to SPEs,
we find empirical evidence, consistent with~\citet{johnson2020platform}, that the seller learning dynamics invariably converge to this 
equilibrium. As such, this  provides useful theoretical support for adopting the family of threshold platform rules by the platform learner. 
We have the following result:
\begin{restatable}{proposition}{propnocollusion}\label{prop:no_collusion}
	For any  $\epsilon > 0$, there exists a threshold platform rule $\pi$ such that	$\mathit{CS}(\pi)\ge \mathit{CS}(\pi^*)-\sum_t\delta^t \epsilon$ under the unique subgame perfect Nash equilibrium (SPE) of any finitely-repeated continuous pricing game induced by platform rule $\pi$.
\end{restatable}
\begin{proof}
For any $\eta>0$, we study the stage game with continuous prices induced by the threshold platform rule that  sets threshold $\tau = c+\eta$ for each price profile $p$. We show that this stage game has a unique Nash equilibrium in which each seller sets a price $p_{i}=p^*$, for some $p^*\in (c,c+\eta]$.  
Given this, we have that every seller pricing at $p^*\in (c,c+\eta]$ in every period is the unique SPE of any finitely repeated game between sellers, since this is an open-loop Nash equilibrium profile (and thus SPE by the single-deviation principle). Moreover (as it is a continuous function over price profiles) the consumer surplus comes arbitrarily close, for a small enough $\eta>0$, to the maximum consumer surplus, which coincides with sellers pricing at cost.

Left to prove is that every seller pricing at $p^*$ is a Nash eq (NE) of the stage game. First, pricing $p_i>\tau(p_t) = c+\eta$ provides zero profit to a seller because the seller is not part of the displayed set of sellers. Similarly, pricing $p_i=c$ provides zero profit.  
Now dropping the time period $t$, because we study a generic stage game, and considering prices $p =(p_1,\ldots,p_n) \in \times_{i=1}^n (c, c+\eta]$, {so that all sellers are displayed},  
and with seller profit $\rho_{i}(p; \mathcal{N})$ for  price profile $p$, we have  
\begin{equation*}
	\frac{\partial}{\partial p_{i}} \rho_{i}(p; \mathcal{N}) = - (p_{i} - c) \frac{D_i(p; \mathcal{N})(1-D_i(p; \mathcal{N}))}{\mu} + D_i(p; \mathcal{N}).  
\end{equation*}
By first-order optimality conditions, we have	$\frac{\partial}{\partial p_{i}}\rho_{i}(p; \mathcal{N})=0$, for each $i$, only when $p_i = \hat p$,  for some $\hat p>c$~\citep{anderson1992logit}. Furthermore, $\rho_{i}(p_{i},p_{-i}; \mathcal{N})$ is concave. Thus, if $c+\eta\ge \hat p$, there is only one Nash equilibrium of the stage game, where each seller $i$ sets price $p^*=\hat p\le c+\eta$. On the other hand, if $c+\eta< \hat p$, we have that $\rho_{i}(p; \mathcal{N})$ is strictly increasing when $p\in \times_{i=1}^n (c, c+\eta]$, and  there is a unique Nash equilibrium where each seller quotes price $p^* = c+\eta$. 
\end{proof}
This proposition follows from a platform rule with a limiting threshold that is arbitrarily close to the sellers' cost $c$. Under this rule, sellers are displayed only if their price is minimal.  At the same time, this particular threshold rule is  fragile, and would lead to market failure if seller costs vary in unexpected ways. By letting the threshold $\tau$ also vary with the price profile $p_t$, as is allowed by the  family of threshold platform rules, we seek to learn milder interventions that still mitigate collusion but remain robust to variations in the  costs faced by sellers in the marketplace. 

\section{Experimental Results}
\label{sec:experiments}

In this section, we evaluate our learning approach via three main experiments. We first consider  performance in terms of consumer surplus, benchmarking our RL interventions against the ones introduced by \citet{johnson2020platform}. We demonstrate the ability to learn optimal leader strategies in the Stackelberg game with the followers across all the seeds we tested, significantly outperforming existing interventions. We then train platform rules ``in the wild," i.e., without access to the sellers' private information, and show that this is not crucial for our learning performances. We conclude by testing the robustness of our interventions,  adding price perturbation during training and evaluating the effect on the robustness of our learned platform rules in environments where sellers have different costs from those assumed during training. 

\paragraph{Experimental set-up.} As in \citet{calvano2020artificial} and \citet{johnson2020platform}, we consider settings with two pricing agents with cost $c=1$, quality indexes $\alpha_1=\alpha_2=2$, and $\alpha_0=0$, and we set parameter $\mu = 0.25$ to control horizontal differentiation. 
The seller Q-learning algorithms are also trained using discount factor $\delta=0.95$, exploration rate $\varepsilon_t = e^{-\beta t}$ with $\beta=1e-5$, and learning rate $\alpha = 0.15$.
We also include results for variations of this default setting in Appendix~\ref{app:setting_variation}.

We adopt five  prices choices for the action of each seller. As in \citet{johnson2020platform}, these prices range from 0.95 ({just below} the sellers' cost) to 2.1 (which is above the monopoly price under no intervention). We note that earlier work provided sellers with a choice of fifteen different prices (over a similar range). We need a smaller grid in order to satisfy our computational constraints; earlier work studied  the effect of different, hand-designed platform rules, and did not also use RL for the automated design of suitable platform rules. We also follow the choices of earlier work, and study an economy with two sellers (again, for reasons of computational resources).
%
%
This coarsened price grid allows us to train a platform policy through Stackelberg POMDP for 50 million steps in 18 hours using a single core on a Intel(R) Xeon(R) Platinum 8268 CPU @ 2.90GHz machine.  

\paragraph{Learning algorithm.} To train the platform policy, we start from the A2C algorithm provided by {\em Stable Baselines3}~\citep[][MIT License]{stable-baselines3}. Given that our policy is only rewarded at the end of a Stackelberg POMDP episode, we configure A2C so that neural network parameters are only updated after this reward phase. In this way, we guarantee that policies inducing desired followers' equilibria are properly rewarded.
Furthermore, to reduce variance in sellers' responses due to non-deterministic policy behavior, we maintain an observation-action map throughout each episode. When a new observation is encountered during the episode, the policy chooses an action following the default training behavior and stores this new observation-action pair in the map. We will show the importance of this variation via an ablation study that is presented in Appendix~\ref{app:deterministic_policies}.
Sellers restart the Q-learning process by re-initializing exploration rates every time the platform rules change (i.e., at the beginning of every Stackelberg POMDP episode). We also show how the training approach is robust to different sellers' behavior models in the Appendix, where the sellers restart the learning rate asynchronously, and not necessarily at the beginning of episodes.

\subsection{Platform Learning Performance}\label{sec:learning_performance} 

In this section, we evaluate the performance of our learned platform policies. For this, we train our policies for 50 million steps in total. We set up the Stackelberg POMDP environment using 50k equilibrium steps and 30 reward steps. These parameters are selected using the following criteria: First, we want to make sure that the equilibrium phase of our Stackelberg POMDP is long enough such that sellers learn to best respond to platform policies. At the same time, we want to avoid too long episodes, as this would lead to few updates of the neural network parameters. With 50k steps we have a good trade-off between these two desiderata. Regarding the reward phase, we want to make sure that {this reward is representative of the converged policy reached by Q-learners, and their converged behavior is usually a single price profile or an loop of two or three price profiles}. We adopt 30 reward steps to be conservative.\footnote{Given that our exploration rate $\varepsilon_t$ is still high after 50k equilibrium steps ($\approx$ 0.6), we artificially set it to zero during the reward phase. We note that this is only a feature to speed up computation and does not violate our ``in the wild'' assumptions, as one could achieve the same effect by using longer equilibrium phases and training runs.}

In these initial experiments, we train our policies using the centralized training-decentralized execution paradigm as used for this Stackelberg learning problem by~\citet{brerolearning}, giving the critic network access to the sellers' learning information (i.e., Q-tables and exploration rates). We relax this below in studying the robustness of the computational framework to online training (``in the wild.") 
We consider the following interventions on behalf of the platform designer:

\begin{itemize}
	\item \textit{No intervention:}  Sellers are always displayed, no matter the price they quote. To derive this baseline, we run our Q learning dynamics until convergence (as described in \citet{johnson2020platform}) for each seed and then average the surplus at final strategies.
	\item \textit{PDP:} We test \textit{price-directed prominence}, a platform intervention introduced by \citet{johnson2020platform}. Here, the platform only displays the seller who quotes the lower price (breaking ties at random), thus enhancing competition. As for \textit{no intervention}, we compute the performance of \textit{PDP} by averaging consumer surplus after Q-learning dynamics converge.
	\item \textit{DPDP:} \textit{Dynamic price-directed prominence} is another  intervention introduced by \citet{johnson2020platform}, which also conditions the choice of the (unique) displayed seller on past prices. Under this intervention, quoting prices equal to cost is a subgame perfect equilibrium of the induced game (under suitable discount factors). 
	As for the previous baselines, we compute the performance of \textit{DPDP} by averaging consumer surplus after Q-learning dynamics converge. 
	
	\item \textit{No State RL:} Here we use the Stackelberg POMDP methodology to train a platform policy that does not use prices $p_t$ to determine the threshold at which to admit each seller to the buy box (thus, ``no state").\footnote{This class of policies already includes the  optimal policy described in the proof of Proposition~\ref{prop:no_collusion}.} Here, Q-learning is restarted whenever a Stackelberg POMDP episode begins.	

    \item \textit{No Stackelberg No State RL:} A variation on ``No State RL"  that does not use the Stackelberg POMDP methodology. Rather, the platform and sellers each follow decentralized learning, and the platform receives a consumer surplus reward at every step. Q-learning is restarted after the same number of steps that are used in a Stackelberg POMDP episode.

	\item \textit{State-based RL:} Here we use the Stackelberg POMDP methodology to train a platform policy that sets a threshold at which to admit each seller as a function of the  price profile quoted by the sellers (thus, ``state-based"). This is the full class of threshold platform rules. Here, Q-learning is restarted whenever a Stackelberg POMDP episodes begins.
	
	\item \textit{No Stackelberg State-based RL:}  A variation on ``State-based RL"  that does not use the Stackelberg POMDP methodology. Rather, the platform and sellers each follow decentralized learning, and the platform receives a consumer surplus reward at every step. Q-learning is restarted after the same number of steps that are used in a Stackelberg POMDP episode.
	
\end{itemize}

\begin{figure}[t]
\begin{center}
\subfigure{
\includegraphics[width=5.8cm]{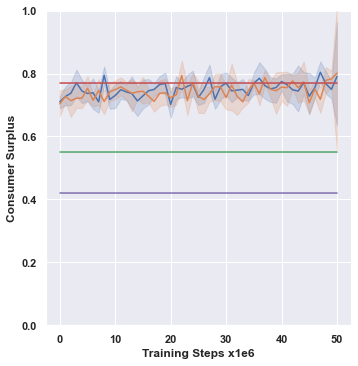}}
\subfigure{
\includegraphics[width=7.6cm]{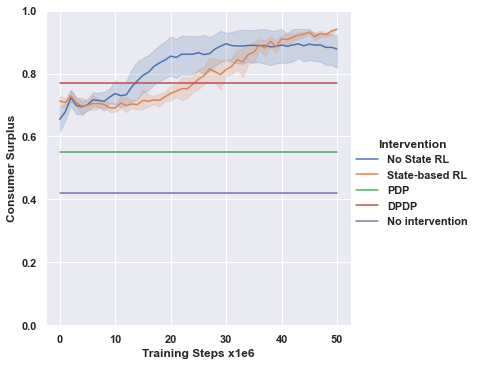}}
\end{center}
\caption{Learning performance of No State RL and State-based RL  compared with different baselines. The results are averaged over 10 runs and shaded regions show 95\% confidence intervals. The No-Stackelberg interventions are displayed on the left, the Stackelberg ones are on the\label{fig:learning_performance2} right.}
\end{figure}
Figure~\ref{fig:learning_performance2} shows the consumer surplus that is realized under these different interventions. 
First, we confirm the results of~\citet{johnson2020platform}, and see consumer surplus improvements from both \textit{PDP} and \textit{DPDP} compared to \textit{No intervention}, with \textit{DPDP} outperforming \textit{PDP}. 
At the same time, the no Stackelberg baselines are not able to outperform \textit{DPDP}, confirming the benefits of using learning methodologies that exploit the leader-follower structure of our game. Indeed, our RL interventions based on the Stackelberg framework dramatically improve consumer surplus, driving it close to its maximal level, which is slightly above 0.94. This is confirmed by the fact that, for both \textit{No State} and \textit{State-based} RL, all sellers are displayed and they invariably quote  minimum prices at the end of training. This is the optimal (i.e., surplus maximizing) seller behavior, confirming  the effectiveness of the Stackelberg-based learning methodology in finding an optimal leader strategy given the Q-learning behavior of sellers. 

\subsection{Learning in the Wild}\label{sec:wild} 
\begin{figure}
\begin{center}
\subfigure{
\includegraphics[width=5.7cm]{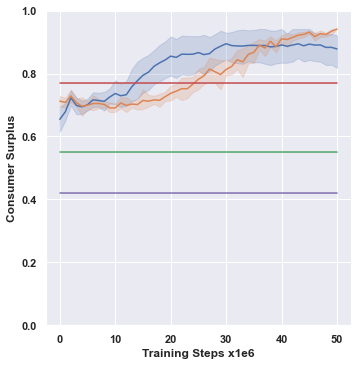}}
\subfigure{
\includegraphics[width=7.7cm]{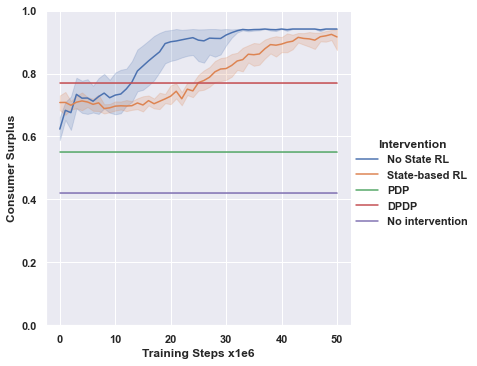}
}
\end{center}
\caption{Offline learning (left) vs.~online, ``learning in the wild" (right) performance. The results are averaged over 10 runs, and the shaded regions show 95\% confidence intervals.}\label{fig:QvsnoQ}
\end{figure}
We now test the performance of the Stackelberg POMDP learning methodology when it has no access to  sellers' private information during training. This can potentially create learning instabilities given that  actor-critic training such as A2C  generally require that the environment accessed by their critic networks is Markovian \citep{grondman2012survey}. Despite this, we find success in this test of ``in the wild" learning. The results are displayed in Figure~\ref{fig:QvsnoQ} and show that, despite relaxing this Markov assumption, the A2C algorithm is able to learn optimal policies for both policy classes (No State and State-based). For No State RL this comes along with lower variance. For State-based RL, the empirical performance is roughly unaffected. We conjecture that the reason behind this good performance is related to the class of threshold platform policies. 
Given a threshold policy, it is possible to predict the overall episode reward based only on the action taken by the policy {(the threshold)} and ignoring the part of the state  that is internal to the sellers (i.e., the Q-matrix and exploration rate). 
Thus, we see empirically that the  additional  information that relates to  sellers' learning can actually make the platform's learning problem harder. 

In Appendix~\ref{app:rewardVariation}, we also demonstrate successful experimental results when we  replace the use of consumer surplus~\eqref{eq:platform_obj} for reward with a reward that corresponds to  the number of agents displayed and the sum of the negated prices offered by sellers. This shows robustness to a possible knowledge gap in knowing the specific functional form of consumer surplus.

\subsection{Robustness of Learned Platform Rules}\label{sec:robustness}

As observed in our previous experiments, the Stackelberg-based RL algorithm is  effective in learning interventions that maximize consumer surplus for a given economic setting. However, as they are tailored to the economic setting at hand, these interventions can perform  poorly when facing settings that are different from those during training. To learn more robust platform rules, we also train with a modified version of the Stackelberg POMDP: at each reward step, with some \textit{random-price probability}, sellers quote prices sampled uniformly at random from the price grid. 
In this way, the platform is  rewarded during training for  performance that remains robust to prices that are not produced by the Q-learning equilibrium dynamics (given seller costs at training). 

We evaluate the effect of adding this perturbation-based robustness to the training procedure in settings with different seller costs: in addition to the default  $c=1.0$, we also test with cost $c=1.38$ (between the second and the third price in our price grid) and  cost $c=1.67$ (between the third and the fourth price in the price grid). Here, and for additional realism, we also continue to train the platform rule  according to the ``in the wild'' approach described above, in Section~\ref{sec:wild}.

As we see in Figure~\ref{fig:robusteness_test}, this  training approach (and in particular with probability 0.4 of random-price perturbation) succeeds in making the state-based policy much more robust in the face of sellers who experience a different cost environment at test time. The robust, state-based policy displays sellers with higher prices (due to their higher costs), while continuing to significantly mitigate collusion when seller costs are as they were during training. This is also confirmed by the policy visualizations in Figure~\ref{fig:policy_visualization}, which show how the buy box learned for State-based RL tends to be much more open under this modified training regime. In contrast, the policy learned by No State RL  performs very poorly (zero consumer surplus) when tested at costs that differ from those assumed during training, and even under this modified training regime. 
There is no single threshold that provides a good compromise between performance at cost $1$ and handling price perturbations. 

\begin{figure}[t] 
\begin{center}
\includegraphics[width=7cm]{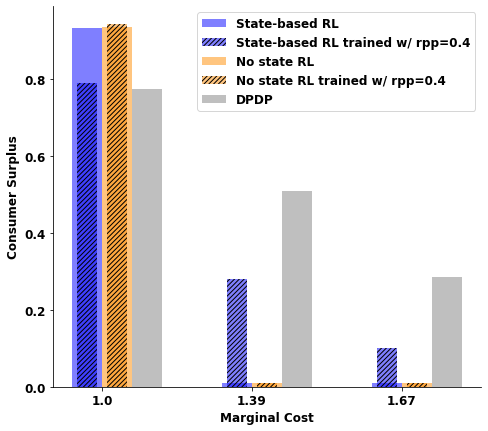}
\end{center}
\caption{{Robustness test}, with   buy box  policy trained without price perturbation and  with  price perturbation with prob.~0.4, averaged over 10 runs. 
}\label{fig:robusteness_test}
\end{figure}

\begin{figure}[h!]
\centering
\subfigure{
\includegraphics[width=6.3cm]{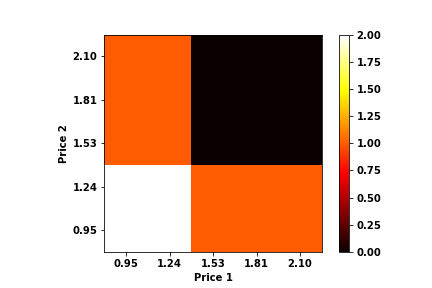}
}\hspace{-1.5cm}
\subfigure{
\includegraphics[width=6.3cm]{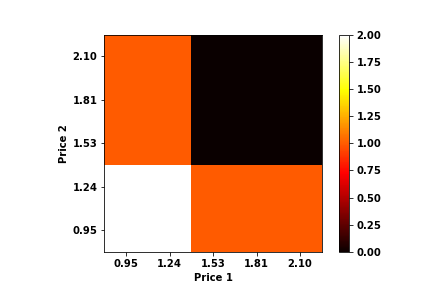}
}\vspace{-0.7cm}
\subfigure{
\includegraphics[width=6.3cm]{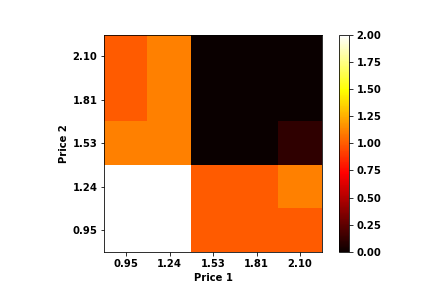}
}\hspace{-1.5cm}
\subfigure{
\includegraphics[width=6.3cm]{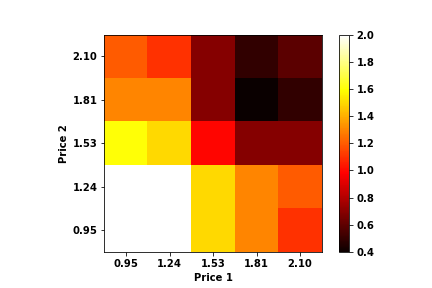}
}
\caption{Policy visualization - number of displayed agents given price selection, averaged over 10 seeds (white--avg.~num.~sellers displayed  2,  black--avg.~num.~sellers displayed 0). (Top-left): No state RL with no price perturbation in the reward step. (Top-right): No state RL with 40\% random price probability in the reward step. (Bottom-left): State-based RL with no price perturbation in the reward step. (Bottom-right): State-based RL with 40\% random price probability in the reward step.}\label{fig:policy_visualization}
\end{figure}

\section{Conclusion}

This work has demonstrated that rules that are effective in preventing collusion by sellers  can be learned through a framework that correctly solves the two-level, Stackelberg problem  (making use of the platform's commitment power). Specifically, we have introduced the class of \textit{threshold policies} that contain policies that optimize consumer surplus and a learning methodology that is effective in learning optimal leader policies in this class. The interventions we learned are shown to substantially outperform the hand-designed interventions introduced in prior work  when the cost environment at test time is as anticipated during training. We also showed how our learned platform interventions can be made more robust when settings are dynamic, with varying seller cost structures, by adopting a  suitably-modified training methodology. This also highlights the importance of the state-based platform rule relative to a no-state rule.
%

Interesting future directions include testing our approach in more complex settings, e.g., when sellers' costs vary between training episodes. In this case, optimal policy actions are based on the prices quoted during the sellers' equilibration phase, as these prices may provide useful information about the current underlying costs (intuitively, the quoted prices will be higher under higher costs). In this scenario, it may be necessary to represent our platform policies via a recurrent neural network, keeping a memory of past prices. 
Finally, we believe that this approach can also be effective in other applications, e.g., to design and understand effective interventions for the electricity markets studied by~\citet{abada2020artificial},   a setting where the successful use of RL  as a defensive response by a platform is not yet established.

\bibliographystyle{ACM-Reference-Format} 
\bibliography{AI_collusion}


\newpage

\appendix

\section{Q Learning}\label{app:Qlearning}
Q learning is an RL algorithm that learns an  estimate of the \emph{action-value function} $Q^*(s, a)$. This action-value function gives the expected total reward of taking action $a$ at state $s$ and  using the optimal policy function in the future.  Once $Q^*(s, a)$ has been learned, the  optimal policy is 
\begin{equation}
\pi^*(s) \in \text{arg max}_{a\in A} Q^*(s, a).
\end{equation}

A Q-learning algorithm maintains an $|S|\times |A|$ Q-matrix, $Q_t$, representing {the estimate of} $Q^*(s, a)$ at step $t$. Usually, this matrix is randomly initialized. At each step $t$, the algorithm takes 
action $a_t$ that with probability $1 - \varepsilon_t$ is  optimal according to its current Q-matrix $Q_t$, and  with probability $\varepsilon_t$   chosen uniformly at random from the set of avalilable actions. {We call $\varepsilon_t$ the {\em exploration rate}.} The entry $(s_t,a_t)$ of $Q_t$ is  updated based on feedback via a convex combination of its previous value and the  reward $r(s_t,a_t)$ attained from the action  plus the discounted value of the state $s_{t+1}$:
\begin{equation}
Q_{t+1}(s_t,a_t) = (1 - \alpha)Q_t(s_t,a_t) + \alpha [r(s_t,a_t) + \delta \max_{a\in A} Q_t(s_{t+1},a)].
\end{equation}

Parameter $\alpha\in [0,1]$ is the \textit{learning rate}. When the environment is stationary and Markovian, and under suitable assumptions on the learning rate and exploration rate, the Q-matrix is guaranteed to converge in the limit to the action-value function $Q^*(s, a)$ and the policy to the optimal policy.

\section{Variations on Setting Parameters}\label{app:setting_variation}
\begin{figure}[h]
\begin{center}
\subfigure{
\includegraphics[width=5.8cm]{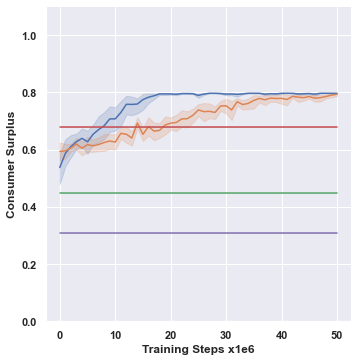}
}
\subfigure{
\includegraphics[width=7.8cm]{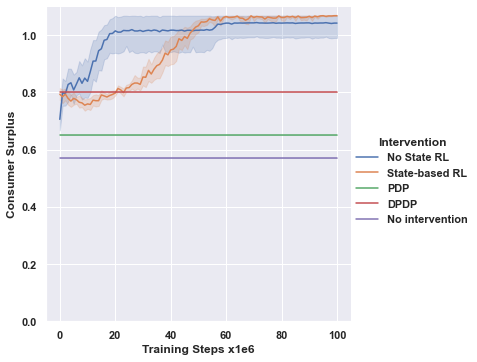}
}
\end{center}
\caption{``Learning in the wild'' performance of No State RL and State-based RL with different environment settings when (left) the horizontal differentiation $\mu$ is lowered from 0.25 to 0.05, and when (right) $\mu$ is increased from 0.25 to 0.40. The results are averaged over 10 runs, and the shaded regions show 95\% confidence intervals.\label{fig:vary-settings}}
\end{figure}
Following \citet{calvano2020artificial} and \citet{johnson2020platform}, we also consider the learning performance of our RL algorithms in settings with different horizontal differentiation parameter $\mu$. Specifically, Figure~\ref{fig:vary-settings} displays our training curves when the horizontal differentiation is lowered from 0.25 to 0.05 and increased from 0.25 to 0.40, and our policies are learned ``in the wild.'' We notice that, in both scenarios, our RL policies bring consumer welfare close to its optimal level (0.8 when $\mu=0.05$ and 1.08 when $\mu=0.4$) outperforming all the baselines.

\section{Deterministic Policies During Stackelberg POMDP Episodes}\label{app:deterministic_policies}
\begin{figure}[h]
\begin{center}
\subfigure{
\includegraphics[width=5.8cm]{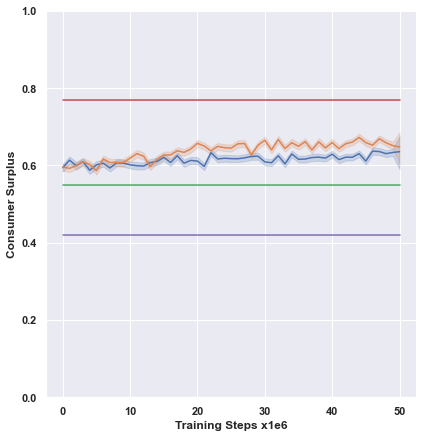}
}
\subfigure{
\includegraphics[width=7.6cm]{"Figures/Stack_wild_performance"}
}
\end{center}
\caption{``Learning in the wild" performance of No State RL and State-based RL when policies are non-deterministic (left) and deterministic (right) during Stackelberg POMDP episodes. The results are averaged over 10 runs and shaded regions show 95\% confidence intervals. \label{fig:nonDetPolicies}}
\end{figure}
In this section, we test the learning performance of our RL algorithms when policies are not deterministic during the Stackelberg POMDP episodes. As we can see from Figure~\ref{fig:nonDetPolicies}, this variation dramatically affects our learning performance, which only slightly improves during training and leads to final policies that do not outperform our baselines. As discussed in Section~\ref{sec:experiments}, this poor performance is caused by the sellers not being able to learn optimal response strategies due to the high variance introduced by the non-deterministic behavior of our policies.

\section{Variations on Sellers' Learning Behavior}
\begin{figure}[h]
\begin{center}
\includegraphics[width=9cm]{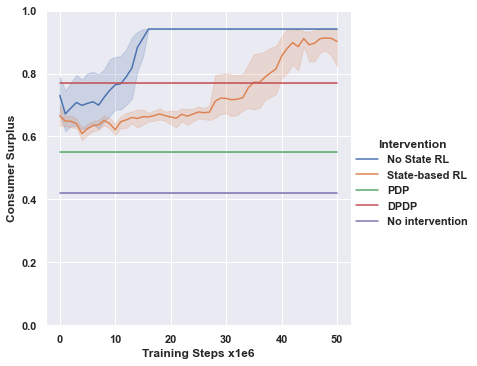}
\end{center}
\caption{``Learning in the wild" performance of No State RL and State-based RL when sellers restart their exploration asynchronously during Stackelberg POMDP episodes. The results are averaged over 10 runs and shaded regions show 95\% confidence intervals. 
\label{fig:asyncRestart}}
\end{figure}
In our previous experiments, we assumed that both sellers restart their learning processes any time the platform rules change. This is consistent with the original experiments run by~\citet{calvano2020artificial}, which demonstrated  seller collusive behavior. However, this assumption may not hold in real-world settings, where sellers can restart their learning processes asynchronously and at any time. This behavior can present new challenges to learning an effective platform policy. Indeed, in this scenario, changes in the sellers' behavior may  be caused not only by different platform interventions, but also as a result of learning restarts. 

In this section, we evaluate the performance of the Stackelberg POMDP framework in scenarios where sellers randomly restart their exploration rate during training. Specifically, we assume that, at each step of the platform's learning process, each seller restarts its exploration rate with some probability. We set this probability such that, in expectation, each seller restarts its exploration once per Stackelberg POMDP episode (which corresponds to the number of steps between platform updates). {A particular concern is that this behavior may result in effectively random prices if a restart occurs close to the nominal reward phase of the Stackelberg POMDP episode. A sensible response of a platform to this would be to monitor sellers' prices, and isolate stages when price profiles are more stable to audit rewards. For a stylized version of this, we allow sellers' exploration to restart as outlined above, but pause any exploration during the reward phase of the Stackelberg POMDP.} Furthermore, if restarts occur close to the reward phases, rewards may reflect an out-of-equilibrium behavior of the sellers (even if exploration is paused). To avoid this problem, we generate our plots by logging rewards in an evaluation Stackelberg POMDP episode we run every 100k training steps. These evaluation episodes use the current platform policy and operate it executing the action with the highest weight given each observation. In these episodes, the Q learning processes are run as in the previous sections, without intra-episode restarts.

As we can see from Figure~\ref{fig:asyncRestart}, the Stackelberg learning framework allows us to derive close-to-optimal policies even under this less stationary behavior. Given that rewards are collected in evaluation episodes (where policies are operated via highest-weighted actions), the optimal intervention under No State RL is executed much earlier than in the simulations of Figure~\ref{fig:QvsnoQ}, reaching the maximum reward after only 15M training steps. 

\section{Variations on Policy Rewards}\label{app:rewardVariation}
In our main experiments, we assume that the platform can compute the consumer surplus $U(p_t; \mathcal{N}_t) = \mu \cdot \log[\lambda(p_t;\mathcal N_t)]$, where $\lambda(p_t;\mathcal N_t) = \sum_{j\in \mathcal N_t} \mathrm{exp}((\alpha_j-p_{j,t})/\mu) +  \mathrm{exp}(\alpha_0/\mu)$ at each step $t$ to reward the policy. We note that, however, to compute this surplus one needs to access the consumers' quality indexes $\alpha_i$s, which may not be available to the platform. However, we note that, to maximize consumer surplus, we can replace $U(p_t; \mathcal{N}_t)$ with any reward function that increases with respect to the number of agents displayed and decreases as prices increase. In this section, we will use $\tilde U(p_t; \mathcal{N}_t) = \mu \cdot \log[\tilde \lambda(p_t;\mathcal N_t)]$, where $\tilde \lambda(p_t;\mathcal N_t) = \sum_{j\in \mathcal N_t} \mathrm{exp}(-p_{j,t}/\mu)$. As we can see from Figure~\ref{fig:varPolicyReward}, our Stackelberg training framework can derive optimal policies even under these modified rewards.
\begin{figure}
\begin{center}
\includegraphics[width=9cm]{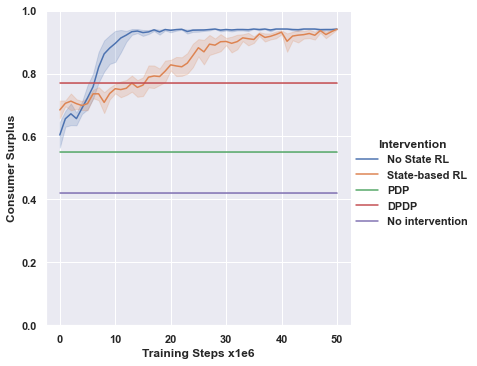}
\caption{Learning performance under reward that does not access quality indexes. \label{fig:varPolicyReward}}
\end{center}
\end{figure}
 
\end{document}